%% file: postproceedings.tex
\documentclass[copyright,creativecommons]{eptcs}

\newcommand{\axrls}{\mathsf{AX}}
\newcommand{\ccrls}{\mathsf{CC}}
\input{macros}

\title{Wave-Style Token Machines\\ and Quantum Lambda Calculi}

\author{Ugo Dal Lago
\institute{Universit\`a di Bologna, Italy \& INRIA}
\and 
Margherita Zorzi
\institute{ Universit\`a di Verona, Italy}
}
\def\titlerunning{Wave-Style Token Machines and Quantum Lambda Calculi}
\def\authorrunning{U. Dal Lago \& M. Zorzi}

\begin{document}
\maketitle

\begin{abstract}
  Particle-style token machines are a way to interpret proofs and
  programs, when the latter are written following the principles of
  linear logic. In this paper, we show that token machines also make
  sense when the programs at hand are those of a simple quantum
  $\lambda$-calculus with implicit qubits. This, however, requires
  generalising the concept of a token machine to one in which more
  than one particle travel around the term \emph{at the same
    time}. The presence of multiple tokens is intimately related to
  entanglement and allows us to give a simple operational semantics to
  the calculus, coherently with the principles of quantum computation.
\end{abstract}

\section{Introduction}
One of the strongest trends in computer science is the (relatively
recent) interest in exploiting new computing paradigms which go beyond
the usual, classical one. Among these paradigms, quantum computing
plays an important role. In particular, the quantum paradigm is having
a deep impact on the notion of a computationally (in)tractable
problem~\cite{Shor97}.

Even if quantum computing has catalysed the interest of a quite large
scientific community, several theoretical aspects are still
unexplored. As an example, the definition of a robust theoretical
framework for quantum programming is nowadays still a challenge. A
number of (paradigmatic) calculi for quantum computing have been
introduced in the last ten years. Among them, some functional calculi,
typed and untyped, have been
proposed~\cite{mscs2009,tcs2010,entcs11,SV06,vT04,Zorzi13}, but we are
still at a stage where it is not clear whether one calculus could be
considered \emph{canonical}. Since quantum data have to undergo
restrictions such as no-cloning and no-erasing, it is not surprising
that in most of the cited quantum calculi the use of resources is
controlled. Linear logic therefore provides an ideal framework for
quantum data treatment, since weakening and contraction (to which
linear logic gives a special status) precisely correspond to erasing
and copying via the Curry-Howard correspondence. But linear logic also
offers another tool which has not been widely exploited in the quantum
setting: its mathematical model in terms of operator algebras,
i.e. the Geometry of Interaction (GoI in the following).  Indeed, the
latter provides a dynamical interpretation and a semantic account of
the cut-elimination procedure as a flow of information circulating
into a net structure. This idea can be formulated both as an algebra
of bounded operators on a infinite-dimensional Hilbert
space~\cite{JYG89} or as a token-based
machine~\cite{GA92,Mackie95}. On the one hand, the Hilbert space on
top of which the first formulation of GoI is given is precisely the
canonical state space of a quantum Turing machine~\cite{BerVa97}. On
the other hand, the definition of a token machine provides a
mathematically simpler setting, which has already found a role in this
context~\cite{DLF11,HH11}.

In this paper, we show that token machines are also a model of a
linear quantum $\lambda$-calculus with implicit quantum bits (qubits),
called $\QL$ and defined along the lines of van Tonder's
$\lambda_q$~\cite{vT04}. This allows us to give an operational
semantics to $\QL$ which renders the quantum nature of $\QL$ explicit:
\emph{type derivations become quantum circuits on the set of gates
  occurring in the underlying $\lambda$-term}. This frees us from the
burden of having to define the operational semantics of quantum
calculi in reduction style, which is known to be technically
challenging in a similar setting~\cite{vT04}.  On the other hand, the
power of $\beta$-style axioms is retained in the form of an equational
theory for which our operational semantics can be proved sound.
Technically, the design of our token machine for $\QL$, called
$\IAM{\QL}$, is arguably more challenging than the one of classical
token machines. Indeed, the principles of quantum computing, and the
so-called \emph{entanglement} in particular, force us to go towards
\emph{wave-style} machines, i.e., to machines where more than one
particle can travel inside the program at the same time. Moreover, the
possibly many tokens at hand are subject to synchronisation points,
each one corresponding to unitary operators of arity greater than
$1$. This means that $\IAM{\QL}$, in principle, could suffer from
deadlocks, let alone the possibility of non-termination. We here prove
that these pathological situations can \emph{never} happen.  In the
present paper we also establish a soundness theorem: we state and
prove that the semantics induced by the token machine $\IAM{\QL}$ is
sound with respect to $\QL$'s equational theory, i.e. it is invariant
with respect to term equivalence. The proof, which we only sketch and
which can anyway be found in~\cite{EV}, is not trivial, since our
notion of term has to deal with quantum superposition~\cite{NieCh00}
and is thus non-standard.  Finally, it is mandatory to recall that,
even if the possibility of observing quantum data is a useful and
expressive programming tool, considering a measurement-free calculus
is a theoretically well-founded choice, since measurements can always
been postponed~\cite{NieCh00}. Thus, this is not a limitation when one
addresses computability issues.

The calculus $\QL$ and its token machine $\IAM{\QL}$ are introduced in
Section~\ref{sec:QL} and Section~\ref{sec:QLTM}, respectively. Main
results about $\IAM{\QL}$ are in Section~\ref{sect:mainres}. An
extended version of this paper with more details, proofs and a gentle
introduction to quantum computing is available~\cite{EV}.
\section{The Calculus $\QL$}\label{sec:QL}
An essential property of quantum programs is that quantum data,
i.e. quantum bits, should always be uniquely referenced. This
restriction follows from the well-known \emph{no-cloning} and
\emph{no-erasing} properties of quantum physics, which state that a
quantum bit cannot in general be duplicated nor
canceled~\cite{NieCh00}.  Syntactically, one captures this restriction
by means of linearity: if every abstraction $\abstr{x}{\termone}$ is
such that there is \emph{exactly one} free occurrence of $x$ in
$\termone$, then the substitution triggered by firing \emph{any} redex
is neither copying nor erasing and, as a consequence, coherent with
the just stated principles.  In this Section, we introduce a quantum
linear $\lambda$-calculus in the style of van Tonder's
$\lambda_q$~\cite{vT04} and give an equational theory for it.
This is the main object of study of this paper, and is the calculus for which we will give a wave-style token
machine in the coming sections.
\subsection{The Language of Terms}\label{sec:syntaxql}
Let us fix a finite set $\uset$ of \emph{unitary operators}, each on a
finite-dimensional Hilbert space $\CC^{2^n}$, where $n$ can be
arbitrary. To each such $\unopone\in\uset$ we associate a symbol
$\opone$ and call $n$ the \emph{arity} of $\opone$. The syntactic
categories of \emph{patterns}, \emph{bit constants}, \emph{constants} and
\emph{terms} are defined by the following grammar:
\begin{align*}
  \patone & ::= \varone\ \midd\ \pair{\varone}{\vartwo}; &&\mbox{\emph{patterns}}\\ 
  \cbitone & ::= \ket{b}_n; &&\mbox{\emph{bit constants}}\\ 
  \csone & ::= \cbitone\ \midd\ \opone; &&\mbox{\emph{constants}}\\ 
  \termone,\termtwo & ::= \varone \midd \csone \midd \termone\otimes\termtwo 
     \midd \termone\termtwo\ \midd\ \lambda \patone.\termone.&&\mbox{\emph{terms}}
\end{align*} 
where $n$ ranges over $\NN$, $b$ ranges over $\{0,1\}$, and
$\varone$ ranges over a denumerable, totally ordered set of variables
$\VV$. We always assume that the natural numbers occurring next to
bits in any term $\termone$ are pairwise distinct. This condition, by
the way, is preserved by substitution when the substituted variable
occurs (free) exactly once.  Whenever this does not cause ambiguity,
we elide labels and simply write $\ket{\bitone}$ for a bit
constant. Notice that pairs are formed via the binary operator
$\otimes$.

We will sometime write $\ket{\bitone_1 \bitone_2\ldots \bitone_k}$ for
$\ket{\bitone_1}\otimes\ket{\bitone_2}\otimes\ldots\otimes\ket{\bitone_k}$
(where $\bitone_1,\ldots,\bitone_n\in\{0,1\}$).
We work modulo variable renaming; in other words, terms are
equivalence classes modulo $\alpha$-conversion.  Substitution up to
$\alpha$-equivalence is defined in the usual way.  Observe that the
terms of $\QL$ are the ones of a $\lambda$-calculus with pairs (which
are accessed by pattern-matching) endowed with constants for bits and
unitary operators. We don't consider measurements here, and discuss
the possibility of extending the language of terms in
sections~\ref{sec:relw} and~\ref{sect:conclusions}.
\subsection{Judgements and Typing Rules.}\label{sec:typing}
We want terms to be non-duplicable and non-erasable by construction
and, as a consequence, we adopt a linear type discipline. Formally,
the set of types is defined as follows
$$
\typeone::=\BB\midd\typeone\linmap\typetwo\midd\typeone\otimes\typetwo,
$$ 
where $\BB$ is the ground type of bits.  We write $\BB^n$ for the
$n$-fold tensor product $\overbrace{\BB\otimes\ldots\otimes\BB}^{n
  \mbox{ {\scriptsize times}}}$.  Judgements and environments are
defined as follows:
\begin{varitemize}
\item  
  A \textit{linear environment} $\conone$ is a (possibly
  empty) finite set of assignments in the form $\varone:\typeone$. We
  impose that in a linear environment, each variable $\varone$ 
  occurs \emph{at most} once;
\item
  If $\conone$ and $\contwo$ are two linear environments assigning types
  to distinct sets of variables, $\conone,\contwo$ denotes their union;
\item 
  A \emph{judgement} is an expression 
  $\conone\vdash\termone:\typeone$, where $\conone$ is a linear  environment, 
  $\termone$ is a term, and $\typeone$ is a type in $\QL$. 
\end{varitemize}
\emph{Typing rules} are in Figure~\ref{fig:TypR}.
\begin{figure}
\begin{center}
\fbox{
\begin{minipage}{.97\textwidth}
\vspace{3pt}
  $$
  \urule{}{\jd{\typet{x}{A}}{\typet{x}{A}}}{\mathsf{(a_v)}}\quad
  \urule{}{\jd{\emcon}{\typet{\ket{0}}{\BB}}}{\mathsf{(a_{q0})}}\qquad
  \urule{}{\jd{\emcon}{\typet{\ket{1}}{\BB}}}{\mathsf{(a_{q1})}}\qquad
  \urule{}{\jd{\emcon}{U:\BB^{n}\linmap\BB^{n}}}{(\mathsf{a_U})}
   $$
   $$
   \urule {\jd{\conone,\varone:\typeone}{\typet{\termone}{\typetwo}}}
   {\jd{\conone}{\typet{\abstr{\varone}{\termone}}{\typeone\linmap\typetwo}}}
   {(\mathsf{I}_\linmap^1)} \qquad \urule
   {\jd{\conone,\varone:\typeone,\vartwo:\typetwo}{\typet{M}{\typethree}}}
   {\jd{\conone}{\typet{\abstr{\langle\varone,\vartwo\rangle}{\termone}}
       {(\typeone\otimes\typetwo)\linmap\typethree}}}
   {(\mathsf{I}_\linmap^2)}
    $$
    $$
    \brule {\jd{\conone}{\typet{M}{A\linmap B}}} {\jd{\contwo}{N: A}}
    {\jd{\conone,\contwo}{MN:B}} {(\mathsf{E}_\linmap)} \qquad \brule
    {\jd{\conone}{\typet{\termone}{\typeone}}}
    {\jd{\contwo}{\typet{\termtwo}{\typetwo}}}
    {\jd{\conone,\contwo}{\typet{\termone\otimes\termtwo}{\typeone\otimes\typetwo}}}
    {(\mathsf{I}_\otimes)}
    $$
\vspace{3pt}
\end{minipage}}
\vspace{-2ex}
\end{center}
\caption{Typing Rules.}\label{fig:TypR}
\end{figure}
Observe that contexts are treated multiplicatively and, as a
consequence, variables always appear exactly once in terms. In other
words, a \emph{strictly linear type discipline} is enforced.
\begin{example}[EPR States]\label{ex:epr}
  Consider the term
  $\termone_{\mathit{EPR}}=\abstr{\pair{\varone}{\vartwo}}{\mathit{CNOT}(\pairtens{\mathit{H}\varone}{y})}$.
  $\termone_{\mathit{EPR}}$ encodes the quantum circuit on two input
  qubits which has the ability to produce an \emph{entangled} state
  from any element of the underlying computational basis\footnote{The
    quantum circuit EPR is built out from the unitary gates $H$ (the
    so-called Hadamard gate) and $\mathit{CNOT}$. The unary gate $H$
    is able to create a superposition of elements of the computational
    basis $\ket{0}$ and $\ket{1}$, i.e. a linear combination in the
    form $\frac{1}{\sqrt{2}}(\ket{0}{+}\ket{1})$ or
    $\frac{1}{\sqrt{2}}(\ket{0}{-}\ket{1})$. The binary gate $\mathit{CNOT}$
    negates its second argument, according to the
    value of the first one.  We provide two simple examples of
    entangled and non-entangled quantum states.  The state
    $\ket{\psi}=\frac{1}{\sqrt{2}} \ket{00}+
    \frac{1}{\sqrt{2}}\ket{11}$ is entangled whereas any state
    $\phi=\alpha\ket{00}+\beta\ket{01}$ is not. In fact, it is
    possible to express the latter in the mathematically equivalent
    form $\phi=\ket{0}\otimes(\alpha\ket{0}+\beta\ket{1})$.
    See~\cite{EV} for a gentle introduction to quantum computing
    essential notions.}.  It can be given the type
  $\BB\otimes\BB\linmap\BB\otimes\BB$ in the empty context. The
  following is a type derivation $\tdone_{\mathit{EPR}}$ for it:
  $$
  \urule
      {
        \brule
            {\jd{\emcon}{\mathit{CNOT}}:\BB\otimes\BB\linmap\BB\otimes\BB}
            {
              \brule
                  {
                    \brule
                        {\jd{\emcon}{\mathit{H}}:\BB\linmap\BB}
                        {\jd{\varone:\BB}{\varone:\BB}}
                        {\jd{\varone:\BB}{\mathit{H}\varone:\BB}}
                        {(\mathsf{E}_\linmap)}
                  }
                  {\jd{\vartwo:\BB}{\vartwo:\BB}}
                  {\jd{\varone:\BB,\vartwo:\BB}{\pairtens{\mathit{H}\varone}{\vartwo}}:\BB\otimes\BB}
                  {(\mathsf{I}_\otimes)}
            }
            {\jd{\varone:\BB,\vartwo:\BB}{\mathit{CNOT}(\pairtens{\mathit{H}\varone}{\vartwo}):\BB\otimes\BB}}
            {(\mathsf{E}_\linmap)}
      }
      {
        \jd{\emcon}{\termone_{\mathit{EPR}}}:\BB\otimes\BB\linmap\BB\otimes\BB
      }
      {(\mathsf{I}_\linmap^2)}
      $$
      $\termone_{\mathit{EPR}}$ and $\tdone_{\mathit{EPR}}$ will be used as running examples in the rest of this
      paper, together with the following type derivation $\tdtwo_{\mathit{EPR}}$:
      $$
      \brule
          {\tdone_{\mathit{EPR}}\pof\jd{\emcon}{\termone_{\mathit{EPR}}}:\BB\otimes\BB\linmap\BB\otimes\BB}
          {
            \brule
                {\jd{\emcon}{\ket{0}_1:\BB}}
                {\jd{\emcon}{\ket{1}_2:\BB}}
                {\jd{\emcon}{\ket{0}_1\otimes\ket{1}_2}:\BB\otimes\BB}
                {(\mathsf{I}_\otimes)}
          }
          {\jd{\emcon}{\termone_{\mathit{EPR}}(\ket{0}_1\otimes\ket{1}_2)}:\BB\otimes\BB}
          {(\mathsf{E}_\linmap)}
          $$
\end{example}
If
$\tdone\pof\jd{\conone}{(\abstr{\varone}{\termone})\termtwo:\typeone}$,
one can build a type derivation $\tdtwo$ with conclusion
$\jd{\conone}{\subst{\termone}{\varone}{\termtwo}:\typeone}$ in a
canonical way, and similarly when
$\tdone\pof\jd{\conone}{(\abstr{\pair{\varone}{\vartwo}}{\termone})(\pairtens{\termtwo}{\termthree}):\typeone}$.
This, as expected, is a consequence of the following:
\begin{lemma}[Substitution Lemma]\label{lemma:substlemma}
If
$\tdone\pof\jd{\conone,\varone_1:\typeone_1,\ldots,\varone_n:\typeone_n}{\termone:\typetwo}$
and for every $1\leq i\leq n$ there is
$\tdtwo_i\pof\jd{\contwo_i}{\termtwo_i:\typeone_i}$, then there is a
canonically defined derivation
$\subst{\tdone}{\varone_1,\ldots,\varone_n}{\tdtwo_1,\ldots,\tdtwo_n}$
of
$\jd{\conone,\contwo_1,\ldots,\contwo_n}{\subst{\termone}{\varone_1,\ldots,\varone_n}{\termtwo_1,
    \ldots,\termtwo_n}:\typetwo}$.
\end{lemma}
\begin{proof}
Just proceed by the usual, simple induction on $\tdone$.
\end{proof}
\subsection{An Equational Theory.}\label{sect:equtheo}
The $\lambda$-calculus is usually endowed with notions of reduction or
equality, both centred around the $\beta$-rule, according to which a
function $\abstr{\varone}{\termone}$ applied to an argument $\termtwo$
\emph{reduces to} (or \emph{can be considered equal to}) the term
$\termone\{\termtwo/\varone\}$ obtained by replacing all free
occurrences of $\varone$ with $\termtwo$. A reduction relation
implicitly provides the underlying calculus with a notion of
computation, while an equational theory is more akin to a reasoning
technique. Giving a reduction relation on $\QL$ terms directly,
however, is problematic. What happens when a $n$-ary unitary operator
$U$ is faced with an $n$-tuple of qubits $\ket{b_1\ldots b_n}$?
Superposition should somehow arise, but how can we capture it?

In this section, an equational theory for $\QL$ will be introduced.
In the next sections, we will show that the semantics induced by token
machines is \emph{sound} with respect to it. The equational theory we
are going to introduce will be a binary relation on formal, weighted
sums of $\QL$ terms:
\begin{definition}[Superposed Term] 
  A \emph{superposed term} of type $(\conone,\typeone)$ is a formal
  sum
  $$
  \sutone=\sum_{i=1}^n\kappa_i\termone_i
  $$ 
  where for every $1\leq i\leq n$, $\kappa_i\in\CC$ and there is
  $\tdone_i$ such that that
  $\tdone_i\pof\jd{\conone}{\typet{\termone_i}{\typeone}}$. In this
  case, we write $\jd{\conone}{\sutone:\typeone}$.
\end{definition}
Superposed terms will be denoted by metavariables like $\sutone$ or
$\suttwo$. Please observe that terms in a superposed term
are \emph{uniformly} typed, i.e., they can be given the same type in
the same context. Please, notice that:
\begin{varitemize}
\item 
  If $\tdone\pof\jd{\emcon}{U\ket{\bitone_1\ldots\bitone_k}:\BB^k}$, then
  there is a naturally defined superposed term of type $(\emcon,\BB^k)$
  which is nothing more than the element of $\CC^{2^k}$ obtained
  by applying $\unopone$ to the vector $\ket{\bitone_1\ldots\bitone_k}$
  With a slight abuse of notation, this superposed term will be indicated
  with $\unopone\ket{\bitone_1\ldots\bitone_k}$.
\item
  All the term constructs can be generalised to operators on
  superposed terms, with the proviso that the types match. As an
  example if $\sutone=\sum_{i}\alpha_{i}\termone_{i}$ where
  $\tdone_i\pof\jd{\conone}{\typet{\termone_i}{\typeone\linmap\typetwo}}$
  and $\tdtwo\pof\jd{\contwo}{\termtwo:\typeone}$, then
  $\sutone\termtwo$ denotes the superposed term
  $\suttwo=\sum_{i}\alpha_{i}(\termone_{i}\termtwo)$. Indeed, there exist type
  derivations
  $\tdthree_i\pof\jd{\conone,\contwo}{\termone_i\termtwo:\typetwo}$
  each obtained applying the rule $(\mathsf{E}_\linmap)$ to $\tdone_i$
  and $\tdtwo$.
\end{varitemize}
It is now time to define our equational theory, which will be defined
on superposed terms of the same type. Formally, $\eqterm$ is a binary
relation on superposed terms, indexed by contexts and types. The fact
that $\sutone\eqterm_{\conone,\typeone}\suttwo$ is indicated with
$\jd{\conone}{\sutone\eqterm\suttwo:\typeone}$.  The relation
$\eqterm$ is defined inductively, by the rules in
Figure~\ref{fig:eqth}.
Notice that for each $\conone,\typeone$, $\eqterm_{\conone,\typeone}$
is by construction an equivalence relation.
\begin{figure}
  \fbox{
  \begin{minipage}{.97\textwidth}
    \vspace{5pt}
    \begin{center}
      \textbf{Axioms}
      \vspace{-2ex}
    \end{center}
    $$
    \urule{\jd{\conone}{(\abstr{\pair{\varone}{\vartwo}}{\termone})(\termtwo\otimes\termthree):\typeone}}
          {\jd{\conone}{(\abstr{\pair{\varone}{\vartwo}}{\termone})(\termtwo\otimes\termthree)\eqterm\subst{\termone}{\varone,\vartwo}{\termtwo,\termthree}:\typeone}}{\mathsf{beta.pair}}
    $$
    $$
    \urule{\jd{\conone}{(\abstr{\varone}{\termone})\termtwo:\typeone}}
          {\jd{\conone}{(\abstr{\varone}{\termone})\termtwo\eqterm\subst{\termone}{\varone}{\termtwo}:\typeone}}{\mathsf{beta}}
    \qquad\qquad    
    \urule{\jd{\emcon}{U\ket{\bitone_1\ldots\bitone_k}:\BB^k}}
          {\jd{\emcon}{U\ket{\bitone_1\ldots\bitone_k}\eqterm\unopone\ket{\bitone_1\ldots\bitone_k}:\BB^k}}{\mathsf{quant}}
    $$
    \vspace{5pt}
    \begin{center}
      \textbf{Context Closure}
      \vspace{-2ex}
    \end{center}
    $$
    \urule
        {\begin{array}{c}
            \jd{\conone}{\sutone\eqterm\suttwo:\typeone\linmap\typetwo}\\
            \jd{\contwo}{\termone:\typeone}
          \end{array}
        }
        {\jd{\conone,\contwo}{\sutone\termone\eqterm\suttwo\termone:\typetwo}}
        {\mathsf{l.a}}
   \qquad
   \urule
       {\begin{array}{c}
           \jd{\conone}{\termone:\typeone\linmap\typetwo}\\
           \jd{\contwo}{\sutone\eqterm\suttwo:\typeone}
        \end{array}}
       {\jd{\conone,\contwo}{\termone\sutone\eqterm\termone\suttwo:\typetwo}}
       {\mathsf{r.a}}
   \qquad
   \urule
       {\jd{\conone,\varone:\typeone}{\sutone\eqterm\suttwo:\typetwo}}
       {\jd{\conone}{\abstr{\varone}{\sutone}\eqterm\abstr{\varone}{\suttwo}:\typeone\linmap\typetwo}}
       {\mathsf{in.}\lambda}
   $$
   $$
   \urule
      {\jd{\conone,\varone:\typeone,\vartwo:\typetwo}{\sutone\eqterm\suttwo:\typethree}}
      {\jd{\conone}{\abstr{\pair{\varone}{\vartwo}}{\sutone}\eqterm\abstr{\pair{\varone}{\vartwo}}{\suttwo}:\typeone\otimes\typetwo\linmap\typethree}}
      {\mathsf{in.}\lambda.\mathsf{pair}}
   \qquad
   \urule
        {\begin{array}{c}
            \jd{\conone}{\sutone\eqterm\suttwo:\typeone}\\
            \jd{\contwo}{\termone:\typetwo}
          \end{array}
        }
        {\jd{\conone,\contwo}{\sutone\otimes\termone\eqterm\suttwo\otimes\termone:\typeone\otimes\typetwo}}
        {\mathsf{l.in.tens}}
   $$
   $$
      \urule
       {\begin{array}{c}
           \jd{\conone}{\termone:\typeone}\\
           \jd{\contwo}{\sutone\eqterm\suttwo:\typetwo}
        \end{array}}
       {\jd{\conone,\contwo}{\termone\otimes\sutone\eqterm\termone\otimes\suttwo:\typeone\otimes\typetwo}}
      {\mathsf{r.in.tens}}
   \qquad
      \urule
      {\begin{array}{c}
          \jd{\conone}{\sutone\eqterm\suttwo:\typeone}\\
          \jd{\conone}{\sutthree:\typeone}
       \end{array}}
      {\jd{\conone}{\alpha\sutone+\sutthree\eqterm\alpha\suttwo+\sutthree:\typeone}}
      {\mathsf{sum}}
   $$
    \vspace{5pt}
    \begin{center}
      \textbf{Reflexive, Symmetric and Transitive Closure}
      \vspace{-4ex}
    \end{center}
    \begin{align*}
      \urule
      {\jd{\conone}{\sutone:\typeone}}
      {\jd{\conone}{\sutone\eqterm\sutone:\typeone}}
      {\mathsf{refl}}
      & &&
      \urule
      {\jd{\conone}{\sutone\eqterm\suttwo:\typeone}}
      {\jd{\conone}{\suttwo\eqterm\sutone:\typeone}}
      {\mathsf{sym}}
      &
      \urule
      {\begin{array}{c}
          \jd{\conone}{\sutone\eqterm\suttwo:\typeone}\\
          \jd{\conone}{\suttwo\eqterm\sutthree:\typeone}
       \end{array}}
      {\jd{\conone}{\sutone\eqterm\sutthree:\typeone}}
      {\mathsf{trans}}
    \end{align*}
    \vspace{3pt}
  \end{minipage}}
  \caption{Equational Theory}\label{fig:eqth}
\end{figure}
\begin{example}
As an example, consider the term
$\termone_{\mathit{EPR}}(\ket{0}_1\otimes\ket{1}_2)$ from
Example~\ref{ex:epr}. The equations in the following chain
can all be derived through axioms and context closure rules:
\begin{align*}
\termone_{\mathit{EPR}}(\ket{0}\otimes\ket{1})&\eqterm\mathit{CNOT}(\pairtens{\mathit{H}\ket{0}}{\ket{1}})
     \eqterm\frac{1}{\sqrt{2}}\mathit{CNOT}(\pairtens{\ket{0}}{\ket{1}})+\frac{1}{\sqrt{2}}\mathit{CNOT}(\pairtens{\ket{1}}{\ket{1}})\\
     &\eqterm\frac{1}{\sqrt{2}}\pairtens{\ket{0}}{\ket{1}}+\frac{1}{\sqrt{2}}\mathit{CNOT}(\pairtens{\ket{1}}{\ket{1}})
     \eqterm\frac{1}{\sqrt{2}}\pairtens{\ket{0}}{\ket{1}}+\frac{1}{\sqrt{2}}\pairtens{\ket{1}}{\ket{0}}.
\end{align*}
The context (which is $\emcon$) and the type (which is $\BB^2$) have
been elided for the sake of readability. By rule $\mathsf{trans}$,
we can derive that
$$
\jd{\emcon}{\mathit{EPR}(\ket{0}\otimes\ket{1})\eqterm\frac{1}{\sqrt{2}}\pairtens{\ket{0}}{\ket{1}}+\frac{1}{\sqrt{2}}\pairtens{\ket{1}}{\ket{0}}:\BB^2}.
$$ 
In other words, $\mathit{EPR}$, when fed with $\ket{0}\otimes\ket{1}$,
produces an entangled pair of qubits. The fourth superposed term in
the chain above has the remarkable property of not being
\emph{homogeneous}, i.e., of being the sum of two terms which are not
identical up to the value of bit constants.
\end{example}
Please observe that the equational theory we have just defined can
\emph{hardly} be seen as an operational semantics for $\QL$. Although
equations can of course be oriented, it is the very nature of a
superposed term which is in principle problematic from the point of
view of quantum computation: what is the mathematical nature of a
superposed term? Is it an element of an Hilbert Space?  And if so, of
\emph{which one}?  If we consider a simple language such as $\QL$, the
questions above may appear overly rhetorical, but we claim they are
not.  For example, what would be the quantum meaning of linear
beta-reduction? If we want to design beta-reduction according to the
principles of quantum computation, it has to be, at least, easily
reversible (unless measurement is implicit in it).  Moving towards
more expressive languages, this non-trivial issue becomes more
difficult and a number of constraints have to be imposed (for example,
superposition of terms can be allowed, but only between homogenous
terms~\cite{vT04}).  This is the reason for which promising
calculi~\cite{vT04} fail to be canonical models for quantum
programming languages.  This issue has been faced in literature
without satisfactory answers, yielding a number of convincing
arguments in favour of the (implicit or explicit) classical control of
quantum data~\cite{mscs2009,SV06}. As observed above, our equational
theory permits non-homogeneous superposed terms in a very natural
way.
\section{A Token Machine for $\QL$}\label{sec:QLTM}
In this section we describe an interpretation of $\QL$ type
derivations in terms of a specific token machine called $\IAM{\QL}$.
Before formally defining $\IAM{\QL}$, it is necessary to give some
preliminary concepts.

With a slight abuse of notation, a permutation
$\permone:\{1,\ldots,n\}\rightarrow\{1,\ldots,n\}$ will be often
applied to sequences of length $n$ with the obvious meaning:
$\permone(a_1,\ldots,a_n)=a_{\permone(1)},\ldots,a_{\permone(n)}$.
Similarly, such a permutation can be seen as the \emph{unique} unitary
operator on $\CC^{2^n}$ which sends $\ket{b_1\cdots b_n}$ to
$\ket{b_{\permone(1)}\cdots b_{\permone(n)}}$.  

Suppose given a unitary operator $\unopone\in\uset$ of arity
$n\in\NN$. Now, take a natural number $m\geq n$ and $n$ distinct
natural numbers $j_1,\ldots,j_n$, all of them smaller or equal to
$m$. With $\unopone_m^{j_1,\ldots,j_n}$ (or simply with
$\unopone^{j_1,\ldots,j_n}$) we indicate the unitary operator of arity
$m$ which acts like $\unopone$ on the quantum bits indexed with
$j_1,\ldots,j_n$ and leave all the other qubits unchanged.  In the
following, with a slight abuse of notation, occurrences of types in
type derivations are confused with types themselves. On the other
hand, occurrences of types \emph{inside other types} will be defined
quite precisely, as follows.

\emph{Contexts} (types with a hole) are denoted by metavariables like
$\ctone$ or $\cttwo$. A context $\ctone$ is said to be \emph{a context
  for a type $\typeone$} if $\ctone[\BB]=\typeone$.  \emph{Negative
  contexts} (i.e., contexts where the hole is in negative position)
are denoted by metavariables like $\nconone,\ncontwo$.
\emph{Positive} ones are denoted by metavariables like
$\pconone,\pcontwo$.  An \emph{occurrence} of $\BB$ in the type
derivation $\tdone$ is a pair $\occone=(\typeone,\ctone)$, where
$\typeone$ is an occurrence of a type in $\tdone$ and $\ctone$ is a
context for $\typeone$.  Sequences of occurrences are indicated with
metavariables like $\soccone,\socctwo$ (possibly indexed). All
sequences of occurrences we will deal with do not contain
duplicates. Type constructors $\linmap$ and $\tens$ can be generalised
to operators on occurrences and sequences of occurrences, e.g.
$(\typeone,\ctone)\linmap\typetwo$ is just
$(\typeone\linmap\typetwo,\ctone\linmap\typetwo)$. If a sequence
of occurrences $\soccone$ contains the occurrences $\occone_1,\ldots,\occone_n$,
we emphasise it by indicating it with $\soccone(\occone_1,\ldots,\occone_n)$.

Given (an occurrence of) a type $\typeone$, all positive and negative
occurrences of $\BB$ inside $\typeone$ form sequences called
$\poccs{\typeone}$ and $\noccs{\typeone}$, respectively. These
are defined as follows (where $\cdot$ is sequence concatenation):
\begin{align*}
  \poccs{\BB}&=(\BB,\emct);\\
  \noccs{\BB}&=\emseq;\\
  \poccs{\typeone\tens\typetwo}&=(\poccs{\typeone}\tens\typetwo)\cdot(\typeone\tens\poccs{\typetwo});\\
  \noccs{\typeone\tens\typetwo}&=(\noccs{\typeone}\tens\typetwo)\cdot(\typeone\tens\noccs{\typetwo});\\
  \poccs{\typeone\linmap\typetwo}&=(\noccs{\typeone}\linmap\typetwo)\cdot(\typeone\linmap\poccs{\typetwo});\\
  \noccs{\typeone\linmap\typetwo}&=(\poccs{\typeone}\linmap\typetwo)\cdot(\typeone\linmap\noccs{\typetwo}).
\end{align*} 
\begin{example}
As an example, the
positive occurrences in the type $\BB\linmap\BB\tens\BB$ should be the
two rightmost ones. And, indeed,
\begin{align*}
\poccs{\BB\linmap\BB\tens\BB}
&=(\noccs{\BB}\linmap\BB\tens\BB)\cdot(\BB\linmap\poccs{\BB\tens\BB})\\
&=\emseq\cdot(\BB\linmap\poccs{\BB\tens\BB})\\
&= \BB\linmap\poccs{\BB\tens\BB}\\
&=(\BB\linmap(\poccs{\BB}\tens\BB))\cdot(\BB\linmap(\BB\tens\poccs{\BB}))\\
&=(\BB,\BB\linmap(\emct\tens\BB)),(\BB,\BB\linmap(\BB\tens\emct)).
\end{align*}
Similarly, one can prove that $\noccs{\BB\linmap\BB\tens\BB}=(\BB,\emct\linmap(\BB\tens\BB))$.
\end{example}
For every type derivation $\tdone$, $\bitocc{\tdone}$ is the sequence
of all these occurrences of $\BB$ in $\tdone$ which are introduced by
the rules $(\mathsf{a_{q0}})$ and $(\mathsf{a_{q1}})$ (recall
Figure~\ref{fig:TypR}). Similarly, $\bitval{\tdone}$ is the
corresponding sequence of binary digits, seen as a vector in
$\CC^{2^{|\bitocc{\tdone}|}}$.  Both in $\bitocc{\tdone}$ and in
$\bitval{\tdone}$, the order is the one induced by the natural number
labeling the underlying bit in $\tdone$. 
\begin{example}
Consider the following type derivation, and call it ${\tau}$:
$$
\brule
    {\jd{\emcon}{\ket{0}_2:\BB_1}}
    {\jd{\emcon}{\ket{1}_1:\BB_2}}
    {\jd{\emcon}{\ket{0}_2\otimes\ket{1}_1}:\BB_3\otimes\BB_4}
    {(\mathsf{I}_\otimes)}
$$ 
There are four occurrences of $\BB$ in it, and we have indexed it with
the first four positive natural numbers, just to be able to point at
them without being forced to use the formal, context machinery. Only
two of them, namely the upper ones, are introduced by instances of the
rules $(\mathsf{a_{q0}})$ and $(\mathsf{a_{q1}})$. Moreover, the
rightmost one serves to type a bit having an index (namely $1$)
greater than the one in the other instance (namely $2$). As a
consequence, $\bitocc{\tau}$ is the sequence $\BB_2,\BB_1$. The two
instances introduces bits $0$ and $1$; then
$\bitval{\tdone}=\ket{1}\otimes\ket{0}$.  As another example, one can
easily compute $\bitocc{\tdone_{\mathit{EPR}}}$ and
$\bitval{\tdone_{\mathit{EPR}}}$ (where $\tdone_{\mathit{EPR}}$ is
from Example~\ref{ex:epr}), finding out that both are the empty
sequence.
\end{example}
We are finally able to define, for every type derivation $\tdone$, the
abstract machine $\autom{\tdone}$ interpreting it:
\begin{varitemize}
\item 
  The \emph{states} of $\autom{\tdone}$ form a set
  $\states{\tdone}$ and are in the form
  $\stone=(\occone_1,\ldots,\occone_n,\qrone)$ where:
  \begin{varitemize}
  \item 
    $\occone_1,\ldots,\occone_n$ are occurrences of the type $\BB$
    in $\tdone$;
  \item 
    $\qrone$ is a \emph{quantum register} on $n$ quantum bits, i.e. a
    normalised vector in $\CC^{2^n}$.
  \end{varitemize}
\item 
  The \emph{transition relation}
  $\trans{\tdone}\subseteq\states{\tdone}\times\states{\tdone}$ is
  defined based on $\tdone$, following the rules from
  Figure~\ref{fig:transone}. In the last rule, $\BB$ in the type of
  $U$ is simply denoted through its index, and for every $1\leq k\leq
  m$, $i_k$ is the position of $\BB_k$ in the sequence $\soccone$.
  The transition rules induced by $(\mathsf{I}_\linmap^2)$
  have been elided for the sake of simplicity
  (see~\cite{EV}).
\end{varitemize}
The number of positive (negative, respectively) occurrences of $\BB$
in the conclusion of $\tdone$ is said to be the \emph{output arity}
(the \emph{input arity}, respectively) of $\tdone$.  Suppose, for the
sake of simplicity, that $\tdone$ is a type derivation of
$\jd{\emcon}{\termone:\typeone}$.  An \emph{initial state} for a
quantum register $\qrone$ is a state in the form
$(\noccs{\typeone}\cdot\bitocc{\tdone},\qrone\otimes\bitval{\tdone})$.
Given a permutation $\permone$ on $n$ elements, a \emph{final state}
for a quantum register $\qrone$ \emph{and} $\permone$ is a state in
the form $(\soccone,\qrone)$, where
$\soccone=\permone(\poccs{\typeone})$. A \emph{run} of $\autom{\tdone}$
is simply a finite or infinite sequence $\stone_1,\stone_2,\ldots$ 
of states from $\states{\tdone}$ such that $\stone_i\trans{\tdone}\stone_{i+1}$
for every $i$.
\begin{figure}
\begin{center}
\fbox{
\begin{minipage}{.97\textwidth}
{\scriptsize
\begin{center}
\begin{tabular}{cc}
\begin{minipage}{2cm}
$$
\urule{}{\jd{\typet{x}{\typeone_1}}{\typet{x}{\typeone_2}}}{(\mathsf{a}_{\mathsf{v}})}
$$
\end{minipage}
&
\begin{minipage}{6.5cm}
\begin{eqnarray*}
((\soccone,(\typeone_1,\pconone),\socctwo),\qrone) &\trans{\tdone}& 
  ((\soccone,(\typeone_2,\pconone),\socctwo),\qrone)\\
((\soccone,(\typeone_2,\nconone),\socctwo),\qrone) &\trans{\tdone}& 
  ((\soccone,(\typeone_1,\nconone),\socctwo),\qrone)
\end{eqnarray*}
\end{minipage}
\end{tabular}

\vspace{1pt}

\begin{tabular}{cc}
\begin{minipage}{3.5cm}
$$
\urule
  {\jd{\conone_1,\varone:\typeone_1}{\typet{\termone}{\typetwo_1}}}
  {\jd{\conone_2}{\typet{\abstr{\varone}{\termone}}{\typeone_2\linmap\typetwo_2}}}
  {(\mathsf{I}^1_\linmap)}
$$
\end{minipage}
&
\begin{minipage}{7cm}
$$
\begin{array}{c}
((\soccone,(\typeone_1,\nconone),\socctwo),\qrone) \trans{\tdone} 
  ((\soccone,(\typeone_2\linmap\typetwo_2,\nconone\linmap\typetwo_2),\socctwo),\qrone)\\ \vspace{-7pt} \\
((\soccone,(\typeone_2\linmap\typetwo_2,\pconone\linmap\typetwo_2),\socctwo),\qrone) \trans{\tdone}
  ((\soccone,(\typeone_1,\pconone),\socctwo),\qrone)\\ \vspace{-7pt} \\
((\soccone,(\typetwo_1,\pconone),\socctwo),\qrone) \trans{\tdone} 
  ((\soccone,(\typeone_2\linmap\typetwo_2,\typeone_2\linmap\pconone),\socctwo),\qrone)\\ \vspace{-7pt} \\
((\soccone,(\typeone_2\linmap\typetwo_2,\typeone_2\linmap\nconone),\socctwo),\qrone) \trans{\tdone} 
  ((\soccone,(\typetwo_1,\nconone),\socctwo),\qrone)\\ \vspace{-7pt} \\
((\soccone,(\conone_2,\pconone),\socctwo),\qrone) \trans{\tdone}
  ((\soccone,(\conone_1,\pconone),\socctwo),\qrone)\\ \vspace{-7pt} \\
((\soccone,(\conone_1,\nconone),\socctwo),\qrone) \trans{\tdone} 
  ((\soccone,(\conone_2,\nconone),\socctwo),\qrone)\\ \vspace{-7pt} \\
\end{array}
$$
\end{minipage}
\end{tabular}

\vspace{1pt}

\begin{tabular}{cc}
\begin{minipage}{4.5cm}
$$
\brule
    {\jd{\conone_1}{\typet{\termone}{\typeone_1\linmap\typetwo_1}}}
    {\jd{\contwo_1}{\termtwo:\typeone}_2}
    {\jd{\conone_2,\contwo_2}{\termone\termtwo:\typetwo_2}}
    {(\mathsf{E}_\linmap)}
$$
\end{minipage}
&
\begin{minipage}{7cm}
$$
\begin{array}{c}
((\soccone,(\typeone_2,\pconone),\socctwo),\qrone) \trans{\tdone} 
  ((\soccone,(\typeone_1\linmap\typetwo_1,\pconone\linmap\typetwo_1),\socctwo),\qrone)\\ \vspace{-7pt} \\
((\soccone,(\typeone_1\linmap\typetwo_1,\nconone\linmap\typetwo_1),\socctwo),\qrone) \trans{\tdone} 
  ((\soccone,(\typeone_2,\nconone),\socctwo),\qrone)\\ \vspace{-7pt} \\
((\soccone,(\typeone_1\linmap\typetwo_1,\typeone_1\linmap\pconone),\socctwo),\qrone) \trans{\tdone} 
  ((\soccone,(\typetwo_2,\pconone),\socctwo),\qrone)\\ \vspace{-7pt} \\
((\soccone,(\typetwo_2,\nconone),\socctwo),\qrone) \trans{\tdone} 
  ((\soccone,(\typeone_1\linmap\typetwo_1,\typeone\linmap\nconone),\socctwo),\qrone)\\ \vspace{-7pt} \\
((\soccone,(\conone_2,\pconone),\socctwo),\qrone) \trans{\tdone} 
  ((\soccone,(\conone_1,\pconone),\socctwo),\qrone)\\ \vspace{-7pt} \\
((\soccone,(\conone_1,\nconone),\socctwo),\qrone) \trans{\tdone} 
  ((\soccone,(\conone_2,\nconone),\socctwo),\qrone)\\ \vspace{-7pt} \\
((\soccone,(\contwo_2,\pconone),\socctwo),\qrone) \trans{\tdone} 
  ((\soccone,(\contwo_1,\pconone),\socctwo),\qrone)\\ \vspace{-7pt} \\
((\soccone,(\contwo_1,\nconone),\socctwo),\qrone) \trans{\tdone} 
  ((\soccone,(\contwo_2,\nconone),\socctwo),\qrone)
\end{array}
$$
\end{minipage}
\end{tabular}
\vspace{2pt}
\begin{tabular}{cc}
\begin{minipage}{4.25cm}
$$
\brule
    {\jd{\conone_1}{\typet{\termone}{\typeone_1}}}
    {\jd{\contwo_1}{\typet{\termtwo}{\typetwo_1}}}
    {\jd{\conone_2,\contwo_2}{\typet{\termone\otimes\termtwo}{\typeone_2\otimes\typetwo_2}}}
    {(\mathsf{I}_\otimes)}
$$
\end{minipage}
&
\begin{minipage}{7.25cm}
$$
\begin{array}{c}
((\soccone,(\typeone_2\otimes\typetwo_2,\nconone\otimes\typetwo_2),\socctwo),\qrone) \trans{\tdone} 
((\soccone,(\typeone_1,\nconone),\socctwo),\qrone)\\ \vspace{-7pt} \\
((\soccone,(\typeone_2\otimes\typetwo_2,\typeone_2\otimes\nconone),\socctwo),\qrone) \trans{\tdone} 
((\soccone,(\typetwo_1,\nconone),\socctwo),\qrone)\\ \vspace{-7pt} \\
((\soccone,(\typeone_1,\pconone),\socctwo),\qrone) \trans{\tdone} 
  ((\soccone,(\typeone_2\otimes\typetwo_2,\pconone\otimes\typetwo_2),\socctwo),\qrone)\\ \vspace{-7pt} \\
((\soccone,(\typetwo_1\pconone),\socctwo),\qrone) \trans{\tdone} 
  ((\soccone,(\typeone_2\otimes\typetwo_2,\typeone_2\otimes\pconone),\socctwo),\qrone)\\ \vspace{-7pt} \\
((\soccone,(\conone_1,\nconone),\socctwo),\qrone) \trans{\tdone} 
  ((\soccone,(\conone_2,\nconone),\socctwo),\qrone)\\ \vspace{-7pt} \\
((\soccone,(\contwo_1,\nconone),\socctwo),\qrone) \trans{\tdone} 
  ((\soccone,(\contwo_2,\nconone),\socctwo),\qrone)\\ \vspace{-7pt} \\
((\soccone,(\conone_2,\pconone),\socctwo),\qrone) \trans{\tdone} 
  ((\soccone,(\conone_1,\pconone),\socctwo),\qrone)\\ \vspace{-7pt} \\
((\soccone,(\contwo_2,\pconone),\socctwo),\qrone) \trans{\tdone} 
  ((\soccone,(\contwo_1,\pconone),\socctwo),\qrone)
\end{array}
$$
\end{minipage}
\end{tabular}

\vspace{6pt}

\begin{tabular}{cc}
\begin{minipage}{4.5cm}
$$
\urule{}{\jd{\emcon}{U:\BB_1\otimes\ldots\otimes\BB_m\linmap\BB_{m+1}\otimes\ldots\otimes\BB_{2m}}}{(\mathsf{a}_{\mathsf{U}})}
$$
\end{minipage}
&
\begin{minipage}{7cm}
$$
\begin{array}{c}
(\soccone(\BB_1,\ldots,\BB_m),\qrone)\\
\trans{\tdone}\\
\qquad(\soccone(\BB_{m+1},\ldots,\BB_{2m}),\unopone^{i_1,\ldots,i_m}(\qrone))
\end{array}
$$
\end{minipage}
\end{tabular}

\vspace{8pt}

\end{center}}
\end{minipage}}
\end{center}
\caption{$\IAM{\QL}$\ Transition Rules}\label{fig:transone}
\end{figure}
\begin{example}[A run of $\IAM{\QL}$]\label{ex:runIAMQ} 
  Consider the term $\termone_{\mathit{EPR}}$ and its type derivation
$\tdone_{\mathit{EPR}}$ (see Example~\ref{ex:epr}).  Forgetting about
terms and marking different occurrences of $\BB$ with distinct
indices, we obtain:
$$
{
\urule
{
  \brule
  {\jd{\emcon}{}\BB_{9}\otimes\BB_{10}\linmap\BB_{11}\otimes\BB_{12}}
  {
    \brule
    {
      \brule
      {\jd{\emcon}{}:\BB_{21}\linmap\BB_{22}}
      {\jd{\BB_{23}}{\BB_{24}}}
      {\jd{\BB_{17}}{\BB_{18}}}
      {(\mathsf{E}_\linmap)}
    }
    {\jd{\BB_{19}}{\BB_{20}}}
    {\jd{\BB_{13},\BB_{14}}\BB_{15}\otimes\BB_{16}}
    {(\mathsf{I}_\otimes)}
  }
  {\jd{\BB_{5},\BB_{6}}{{}{}\BB_{7}\otimes\BB_{8}}}
  {(\mathsf{E}_\linmap)}
}
{
  \jd{\emcon}{}\BB_{1}\otimes\BB_{2}\linmap\BB_{3}\otimes\BB_{4}
}
{(\mathsf{I}_\linmap^2)}
}
$$ 
Let us consider the following computation of
$\autom{\tdone_{\mathit{EPR}}}$:
\begin{align*}
  (\BB_1,&\BB_2,\qrone)\trans{\tdone}^*(\BB_{5},\BB_{6},\qrone)\trans{\tdone}^*(\BB_{13},\BB_{14},\qrone)
  \trans{\tdone}(\BB_{17},\BB_{19},\qrone)\trans{\tdone}^*(\BB_{23},\BB_{20},\qrone)\\
  &\trans{\tdone}^*(\BB_{24},\BB_{16})\trans{\tdone}(\BB_{24},\BB_{10},\qrone)\trans{\tdone}(\BB_{21},\BB_{10},\qrone)
  \trans{\tdone}(\BB_{22},\BB_{10},\mathbf{H}^1(\qrone))\\
  &\trans{\tdone}(\BB_{18},\BB_{10},\mathbf{H}^1(\qrone))\trans{\tdone}(\BB_{15},\BB_{10},\mathbf{H}^1(\qrone))
  \trans{\tdone}(\BB_{9},\BB_{10},\mathbf{H}^1(\qrone))\\
  &\trans{\tdone}(\BB_{11},\BB_{12},{\mathbf{CNOT}}^{1,2}(\mathbf{H}^1(\qrone)))\trans{\tdone}^*(\BB_{7},\BB_{8},\mathbf{CNOT}^{1,2}(\mathbf{H}^1(\qrone)))\\
  &\trans{\tdone}(\BB_{3},\BB_{4},\mathbf{CNOT}^{1,2}(\mathbf{H}^1(\qrone))).
\end{align*}
Notice that the occurrence of $\mathit{CNOT}$ acts as a
synchronisation operator: the second token is stuck at the occurrence
$\BB_{10}$ until the first token arrives (from the occurrence
$\BB_{15}$) as a control input of the $\mathit{CNOT}$ and the
corresponding reduction step actually occurs.
\end{example}
What the example above shows, indeed, is that the presence of a
potential \emph{entanglement} in $\tdone$ is intimately related to the
necessity of \emph{synchronisation} in the underlying machine
$\autom{\tdone}$: if all unitary operators in $\tdone$ can be
expressed as the tensor product of unitary operators of arity
one (and, thus, entanglement is not possible), then synchronisation
is simply not necessary.

Given a type derivation $\tdone$, the relation $\trans{\tdone}$ enjoys
a strong form of confluence:
\begin{proposition}[One-step Confluence of $\trans{\tdone}$]
  Let $\stone, \sttwo, \stthree\in \states{\tdone}$ be such that
  $\stone\trans{\tdone} \sttwo$ and $\stone\trans{\tdone}\stthree$.
  Then either $\sttwo=\stthree$ or there exists a state $\stfour$ such
  that $\sttwo\trans{\tdone} \stfour$ and $\stthree\trans{\tdone}
  \stfour$.
\end{proposition}
\begin{proof}
By simply inspecting the various rules. Notice that there are no critical pairs in 
$\trans{\tdone}$. 
\end{proof}
The way $\autom{\tdone}$ is built by following a type derivation $\pi$
induces the following notion:
\begin{definition}
Given a type derivation $\tdone$, the partial function \emph{computed by $\tdone$} 
is denoted as $\pfun{\tdone}$, has domain $\CC^{2^n}$ and codomain $\CC^{2^m}$
(where $n$ and $m$ are the input and output arity of $\tdone$) and is
defined by stipulating that $\pfun{\tdone}(\qrone)=\qrtwo$ iff
any initial state for $\qrone$ rewrites into a 
final state for $\qrthree$ and $\permone$, where $\qrthree=\permone^{-1}(\qrtwo)$. 
\end{definition}
Given a type derivation $\tdone$, $\pfun{\tdone}$ is either always undefined or
always defined. Indeed, the fact any initial configuration (for, say, $\qrone$)
rewrites to a final configuration or not does \emph{not} depend on $\qrone$ but only
on $\tdone$:
\begin{lemma}[Uniformity]
For every type derivation $\tdone$ and for every occurrences
$\occone_1,\ldots,\occone_n$, $\occtwo_1,\ldots,\occtwo_n$, there is a unitary
operator $\unopone$ such that
whenever $(\occone_1,\ldots,\occone_n,\qrone)\trans{\tdone}(\occtwo_1,\ldots,\occtwo_n,\qrtwo)$
it holds that $\qrtwo=\unopone(\qrone)$.
\end{lemma}
\begin{proof}
Observe that for every $\occone_1,\ldots,\occone_n$, $\occtwo_1,\ldots,\occtwo_n$ there
is \emph{at most} one of the rules defining $\trans{\tdone}$ which can be applied. Moreover,
notice that each rule acts uniformly on the underlying quantum register.
\end{proof}
In the following section, we will prove that $\pfun{\tdone}$ is always
a \emph{total} function, and that it makes perfect sense from a
quantum point of view.
\section{Main Properties of $\IAM{\QL}$}\label{sect:mainres}
In this section, we will give some crucial results about
$\IAM{\QL}$. More specifically, we prove that runs of this token
machine are indeed finite and end in final states.  We proceed by
putting $\QL$ in correspondence to $\MLL$, inheriting its very elegant
proof theory and token machines.
\subsection{A Correspondence Between $\MLL$ and $\QL$}\label{sec:mllql}
Let $\AAA=\{\atomone, \atomtwo,\ldots\}$ be a countable set of \emph{propositional atoms}. 
A \emph{formula} $\formone$ of Multiplicative Linear Logic (\MLL) is given by the following grammar:
$$
\formone,\formtwo::= \atomone\midd\lneg{\atomone}\midd\formone\otimes\formtwo\midd\formone\lpar\formtwo.
$$
Linear negation can be extended to all formulas in the usual way: 
$$
  \lneg{(\lneg{\alpha})}=\alpha;\qquad\qquad
  \lneg{\formone\otimes\formtwo}=\lneg{\formone}\lpar\lneg{\formtwo};\qquad\qquad
  \lneg{\formone\lpar\formtwo}=\lneg{\formone}\otimes\lneg{\formtwo}.
$$ 
This way, $\dneg{\formone}$ is just $\formone$.  The one-sided sequent
calculus for \MLL\ is very simple:
$$
  \urule{}{\vdash{\formone,\lneg{\formone}}}{\mathsf{ax}}
  \quad
  \brule{\vdash{\conone,\formone}}{\vdash\contwo,\lneg{\formone}}{\vdash{\conone,\contwo}}{\mathsf{cut}}
  \qquad
  \brule{\vdash{\conone,\formone}}{\vdash\contwo,\formtwo}{\vdash{\conone,\contwo,\formone\otimes\formtwo}}{\otimes}
  \qquad
  \urule{\vdash{\conone,\formone,\formtwo}}{\vdash{\conone,\formone\lpar\formtwo}}{\lpar}
$$ 
The logic \MLL\ enjoys cut-elimination: there is a terminating
algorithm turning any \MLL\ proof into a cut-free proof of the same
conclusion. A notion of \emph{structural equivalence} between two
\MLL\ proofs $\pmllone,\pmlltwo$ having the same conclusion
$\vdash{\conone}$ can be easily defined and holds only if $\pmllone$
and $\pmlltwo$ are essentially \emph{the same} proof modulo
renaming of the formulas occurring in $\pmllone$ and $\pmlltwo$.
Remarkably, two \MLL\ proofs which are structurally equivalent
are actually the \emph{same} proof, a result which does not hold
in more expressive logics like \MELL. More details on that can
be found in~\cite{EV}.

Any $\QL$ type derivation $\tdone$ can be put in correspondence with
\emph{some} \MLL\ proofs.  We inductively define the map
$\mapone{\cdot}$ from $\QL$ types to \MLL\ formulas as follows:
$$
  \mapone{\BB}=\alpha;\qquad\qquad
  \mapone{A\linmap B}=\lneg{{\mapone{A}}}\lpar\mapone{B};\qquad\qquad
  \mapone{A\otimes B}=\mapone{A}\otimes\mapone{B}
$$
Given a judgment $\judgone=\jd{\conone}{\termone:\typeone}$ and a
natural number $n\in\NN$, the \MLL\ sequent \emph{corresponding} to
$\judgone$ and $n$ is the following one:
$$
\vdash\underbrace{\lneg{\alpha},\ldots,\lneg{\alpha}}_{\mbox{$n$ times}},
\lneg{(\mapone{\typetwo_1})},\ldots,\lneg{(\mapone{\typetwo_m})},\mapone{\typeone},
$$ 
where $\conone= x_1:\typetwo_1,\ldots,x_m:\typetwo_m$. For every
$\tdone$, we define now a set of \MLL\ proofs $\maptwo{\tdone}$. This
way, every type derivation $\tdone$ for
$\judgone=\jd{\conone}{\termone:\typeone}$ such that $n$ bits occur in
$\termone$, is put in relation to possibly many \MLL\ proofs of the
sequent corresponding to $\judgone$ and $n$.  One among them is called
the \emph{canonical proof} for $\tdone$. The set $\maptwo{\tdone}$ and
canonical proofs are defined by induction on the structure of the
underlying type derivation $\tdone$. The type constructions of $\QL$
are mapped to the corresponding $\MLL$ logical operators, rules
$\mathsf{(a_{q0})}$ and $\mathsf{(a_{q1})}$ are mapped to axioms, and
rule $(\mathsf{a_U})$ is mapped to a proof encoding a permutation of
the involved atoms.  When the latter is the identity, we get the
{canonical proof} for $\tdone$. For more details, please refer
to~\cite{EV}.

Given an $\mll$ proof $\pmllone$, let us denote as
$\statesmll{\pmllone}$ the class of all finite sequences of atom
occurrences in $\pmllone$.  The relation $\transmll{\pmllone}$ can be
extended to a relation on $\statesmll{\pmllone}$ by stipulating that
$$
(\occone_1,\ldots,\occone_{n-1},\occtwo,\occone_{n+1},\ldots,\occone_{m})
\transmll{\pmllone}
(\occone_1,\ldots,\occone_{n-1},\occthree,\occone_{n+1},\ldots,\occone_{m})
$$ 
whenever $\occtwo\transmll{\pmllone}\occthree$. As usual,
$\transmll{\pmllone}^+$ is the transitive closure of
$\transmll{\pmllone}$.

Let us now consider a type derivation $\tdone$ in $\QL$, its quantum
token machine $\autom{\tdone}$, and any
$\pmllone\in\maptwo{\tdone}$. States of $\autom{\tdone}$ can be mapped
to $\statesmll{\pmllone}$ by simply forgetting the underlying quantum
register and mapping any occurrence of $\tdone$ to the corresponding
atom occurrence in $\pmllone$. This way one gets a map
$$
\mapthreenp{\tdone}{\pmllone}:\states{\tdone}\rightarrow \statesmll{\pmllone}
$$
such that, given a state $\stgen=(\occone_1,\ldots,\occone_n,\qrone)$ in $\states{\tdone}$, 
$|\mapthree{\tdone}{\pmllone}{\stgen}|=n$, i.e., the number of occurrences in $\stgen$ 
is the same as the length of $\mapthree{\tdone}{\pmllone}{\stgen}$.
Each reduction step on the token machine $\autom{\tdone}$ corresponds to \emph{at least one} 
reduction step in the $\mll$ machine $\mmllone{\pmllone}$, where $\pmllone\in\maptwo{\tdone}$ 
is the canonical proof:
\begin{lemma}\label{lemma:corresponding}
  Let us consider a token machine $\autom{\tdone}$ and two states
  $\stone,\sttwo\in\states{\tdone}$. If $\stone\trans{\tdone}\sttwo$
  and $\pmllone\in\maptwo{\tdone}$ is canonical, then
  $\mapthree{\tdone}{\pmllone}{\stone}\transmll{\pmllone}^+\mapthree{\tdone}{\pmllone}{\sttwo}$.
\end{lemma}
\begin{proof}
This goes by induction on the structure of $\tdone$.
\end{proof}
Any (possible) pathological situation on the quantum token machine,
then, can be brought back to a corresponding (absurd) pathological
situation in the $\MLL$ token machine. This is the principle that will
guide us in the rest of this section.
\subsection{Termination, Progress and Soundness}
The first property we want to be sure about is that every computation
of any token machine $\autom{\tdone}$ always terminates. The second
one is progress (i.e. deadlock-freedom). In both cases, we use in an
essential way the correspondence between $\QL$ and $\MLL$.
\begin{proposition}[Termination]~\label{prop:term} 
  For any quantum token machine $\autom{\tdone}$, any sequence
  $\stone\trans{\tdone}\sttwo\trans{\tdone}\ldots$ is finite.
\end{proposition}
\begin{proof}
  Suppose, for the sake of contradiction, than there exists an
  infinite computation in $\autom{\tdone}$. This implies by
  Lemma~\ref{lemma:corresponding} that there exists an infinite path
  in the token machine $\mmllone{\pmllone}$ where $\pmllone$ is the
  canonical \MLL\ proof for $\tdone$. This is a contradiction,
  because paths in \MLL\ proofs are well-known to be always
  finite.
\end{proof}

Progress (i.e. deadlock-freedom) is more difficult to prove than
termination. Given a type derivation $\tdone$, an \emph{argument
  occurrence} is any negative occurrence $(\typeone,\nconone)$ of
$\BB$ in a $(\mathsf{a_U})$ axiom. We extend this definition to the
corresponding atom occurrence when $\pmllone\in\maptwo{\tdone}$.  A
\emph{result occurrence} is defined similarly, but the occurrence has
to be positive.
\begin{proposition}[Progress]\label{prop:progress}
  Suppose $\tdone$ is a type derivation in $\QL$ and
  $\stone\in\states{\tdone}$ is initial. Moreover, suppose that
  $\stone\trans{\tdone}^*\sttwo$. Then either $\sttwo$ is final or
  $\sttwo\trans{\tdone}\stthree$ for some
  $\stthree\in\states{\tdone}$.
\end{proposition}
\begin{proof}
Let us consider a computation
$\stgen_1\trans{\tdone}\ldots\trans{\tdone}\stgen_k$ on a quantum
token machine $\autom{\tdone}$. Suppose that the state $\stgen_k$ is a
deadlocked state, i.e. $\stgen_k$ is not a final state, and that there
exists no $\stgen_{m}$ such that $\stgen_k\trans{\tdone}\stgen_m$.
The fact $\stgen_k$ is a deadlocked state means that $l\geq 1$
occurrences in $\stgen_k$ are argument occurrences, since the latter
are the only points of synchronisation of the machine.  Let us
consider any \emph{maximal} sequence
\begin{equation}\label{equ:maxseq}
  \mapthree{\tdone}{\pmllone}{\stgen_1}\transmll{\pmllone}\ldots\transmll{\pmllone}\mapthree{\tdone}{\pmllone}{\stgen_k}
    \transmll{\pmllone}\mathsf{Q_1}\transmll{\pmllone}\ldots\transmll{\pmllone}\mathsf{Q}_n,
\end{equation}
where $\pmllone\in\maptwo{\tdone}$ is the canonical proof
corresponding to $\tdone$.  Observe that in (\ref{equ:maxseq}), all
occurrences of atoms in $\pmllone$ are visited exactly once, including
those corresponding to argument and result occurrences from
$\tdone$. Notice, however, that the argument and result occurrences of
the unitary operators affected by $\stgen_k$ cannot have been visited
along the subsequence
$\mapthree{\tdone}{\pmllone}{\stgen_1}\transmll{\pmllone}\ldots\transmll{\pmllone}\mapthree{\tdone}{\pmllone}{\stgen_k}$
(otherwise we would visit the occurrences in $\stgen_k$ at least
twice, which is not possible).  Now, form a directed graph whose nodes
are the unitary constants $U_1,\ldots,U_h$ which block $\stgen_k$,
plus a node $F$ (representing the conclusion of $\tdone$), and whose
edges are defined as follows:
\begin{varitemize}
\item
  there is an edge from $U_i$ to $U_j$ iff along
  $\mathsf{Q_1}\transmll{\pmllone}\ldots\transmll{\pmllone}\mathsf{Q}_n$
  one of the $l$ independent computations corresponding to a blocked
  occurrence in $\stgen_{k}$ is such that a result occurrence of $U_i$
  is followed by an argument occurrence of $U_j$ and the occurrences
  between them are neither argument nor result occurrences.
\item
  there is an edge from $U_i$ to $F$ iff along
  $\mathsf{Q_1}\transmll{\pmllone}\ldots\transmll{\pmllone}\mathsf{Q}_n$
  one of the $l$ traces is such that a result occurrence of $U_i$ is
  followed by a final occurrence of an atom and the occurrences
  between them are neither argument nor result occurrences.
\end{varitemize}
The thus obtained graph has the following properties:
\begin{varitemize}
\item
  Every node $U_i$ has at least one incoming edge, because otherwise
  the configuration $\stgen_k$ would not be deadlocked.
\item
  As a consequence, the graph must be cyclic, because otherwise we
  could topologically sort it and get a node with no incoming edges
  (meaning that some of the $U_i$ would not be blocked!). Moreover,
  the cycle does not include $F$, because the latter only has incoming
  nodes.
\end{varitemize}
From any cycle involving the $U_j$, one can induce the presence of a
cycle in the token machine $\mmllone{\pmlltwo}$ for some
$\pmlltwo\in\maptwo{\tdone}$. Indeed, such a $\pmlltwo$ can be formed
by simply choosing, for each $U_j$, the ``good'' permutation, namely
the one linking the incoming edge and the outgoing edge which are part
of the cycle. This way, we have reached the absurd starting from the
existence of a deadlocked computation.
\end{proof}

The immediate consequence of the termination and progress results is
that $\pfun{\tdone}$ is always a \emph{total} function. The way
$\autom{\tdone}$ is defined ensures that $\pfun{\tdone}$ is obtained
by feeding some of the inputs of a unitary operator $\unopone$ with
some bits (namely those occurring in $\tdone$). $\unopone$ is itself
obtained by composing the unitary operators occurring in $\tdone$,
which can thus be seen as a program computing a quantum circuit. In a
way, then, token machines both show that $\QL$ \emph{is a truly quantum
calculus} and can be seen as the right operational semantics for it.

The last step consists in understanding the relation between token
machines and the equational theory on superposed terms introduced in
Section~\ref{sect:equtheo}. First of all, observe that
$\sutone=\sum_{i=1}^n\kappa_i\termone_i$ has type $\typeone$ in the
context $\conone$, then $\termone_1,\ldots,\termone_n$ all have type
$\typeone$ in the context $\conone$. But there is more to that: for
every $1\leq i\leq n$, there is \emph{exactly one} type derivation
$\tdone_i\pof\jd{\conone}{\termone_i:\typeone}$. This holds because
two such type derivations $\tdone_i$ and $\tdtwo_i$ are such that the
canonical proofs in $\maptwo{\tdone_i}$ and $\maptwo{\tdtwo_i}$ are
structurally equivalent, thus identical. It is then possible to extend
the definition of $\pfun{\cdot}$ to superposed terms: if
$\sutone=\sum_{i=1}^n\kappa_i\termone_i$ has type $\typeone$ in
$\conone$, then $\pfun{\sutone}$, when fed with a vector $\vecone$,
returns $\sum_{i=1}^n\kappa_i\pfun{\tdone_i}(\vecone)$, where
$\tdone_i$ is the \emph{unique} derivation giving type $\typeone$ to
$\termone_i$ in the context $\conone$. Remarkably, token machines
behave in accordance to the equational theory: this is our Soundness
Theorem.
\begin{theorem}[Soundness]
Given $\sutone$ and $\suttwo$ superposed terms, if
$\jd{\conone}{\sutone\eqterm\suttwo:\typeone}$, then
$\pfun{\sutone}=\pfun{\suttwo}$.
\end{theorem}
\begin{proof}
We only give a sketch of the proof. More details can been found
in~\cite{EV}. The first step consists in proving that any derivation
of $\jd{\conone}{\sutone\eqterm\suttwo:\typeone}$ can be put in
\emph{normal form}, a concept defined by giving an order on the rules
in Figure~\ref{fig:eqth}. More specifically, define the following two
sets of rules:
\begin{align*}
\axrls&=\{\mathsf{beta},\mathsf{beta.pair},\mathsf{quant}\};\\
\ccrls&=\{\mathsf{l.a},\mathsf{r.a},\mathsf{in.}\lambda,\mathsf{in.}\lambda\mathsf{.pair},\mathsf{l.in.tens},\mathsf{r.in.tens}\}.
\end{align*}
A derivation of $\jd{\conone}{\sutone\eqterm\suttwo:\typeone}$ is said
to be \emph{in normal form} (and we write
$\jd{\conone}{\sutone\eqtermnf\suttwo:\typeone}$) iff
\begin{varitemize}
\item
  either the derivation is obtained by applying rule $\mathsf{refl}$;
\item
  or any branch in the derivation consists in instances of rules from
  $\axrls$, possibly followed by instances of rules in $\ccrls$,
  possibly followed by instances of $\mathsf{sum}$, possibly followed
  by instances of $\mathsf{sym}$ possibly followed by instances of
  $\mathsf{trans}$.
\end{varitemize}
In other words, a derivation of
$\jd{\conone}{\sutone\eqterm\suttwo:\typeone}$ is in normal form iff
rules are applied in a certain order. As an example, we cannot apply
transitivity or symmetry closure rules too early, i.e., before context
closure rules. One may wonder whether this restricts the class of
provable equivalences. Infact it does not:
$\conone{\conone}{\sutone\eqterm\suttwo:\typeone}$ iff 
$\conone{\conone}{\sutone\eqtermnf\suttwo:\typeone}$, a result
which is not particularly deep although a bit tedious to prove~\cite{EV}.
Once we have this result in our hands, however, proving Soundness
becomes much easier, since the difficult and problematic rules,
namely those in $\ccrls$, are applied to superposed
terms of a very specific shape, namely those obtained through
$\axrls$.
\end{proof}
\section{Related Work}\label{sec:relw}
In~\cite{HH11}, a geometry of interaction model for Selinger and
Valiron's quantum $\lambda$-calculus~\cite{SV06} is defined.  The
model is formulated in particle-style.  In~\cite{DLF11} \QMLL, an
extension of \MLL\ with a new kind of modality, is studied. \QMLL\ is
sound and complete with respect to quantum circuits, and an
interactive (particle-style) abstract machine is defined.  In both
cases, adopting a particle-style approach has a bad consequence: the
``quantum'' tensor product does \emph{not} coincide with the tensor
product in the sense of linear logic. Here we show that adopting the
wave-style approach solves the problem. Quantum extensions of
game semantics are partially connected to this work. See, for
example~\cite{DP08,Del11}. Purely linear quantum lambda-calculi
(\emph{with} measurements) can be given a fully abstract denotational
semantics, like the one proposed by Selinger and
Valiron~\cite{SelingerV08}. In their work, closure (necessary to
interpret higher-order functions) is not obtained via traces and is
not directly related in any way to the geometry of
interaction. Moreover, morphisms are just linear maps, and so the
model is far from being a quantum operational semantics. A language of
terms similar to $\QL$ has been also studied in~\cite{Yoshi14}, where
the calculus of proof-nets $\mathsf{MLL}_{\mathsf{qm}}$ is
introduced. $\mathsf{MLL}_{\mathsf{qm}}$'s syntax also includes a
measurement box-like operator (which models the possibility of
``observe'' the value of a quantum bit~\cite{NieCh00}). A multi-token
machine semantics for $\mathsf{MLL}_{\mathsf{qm}}$ proof-nets is
defined and proved to be \emph{sound}, i.e. invariant along reduction
of proof nets. Moreover, although a $\lambda$-calculus is given,
together with a compilation scheme to $\mathsf{MLL}_{\mathsf{qm}}$
proof-nets, the considered $\lambda$-calculus is one with
\emph{explicit} qubits, contrary to $\QL$. Finally, Arrighi
and Dowek's work shows that turning a sum-based algebraic $\lambda$-calculus
into a quantum computational model can be highly non-trivial~\cite{ArrDow}.
\section{Conclusions}\label{sect:conclusions}
We have introduced $\IAM{\QL}$, an interactive abstract machine which
provides a sound operational characterisation of any type derivation
in a linear quantum $\lambda$-calculus $\QL$. This is an example of a
concrete wave-style token machine whose runs cannot be seen simply as
the asynchronous parallel composition of particle-style
runs. Interestingly, synchronisation is intimately related to
entanglement: if, for example, only unary operators occur in a term
(i.e. entanglement is \emph{not} possible), synchronisation is not
needed and everything collapses to the particle-style.  Our
investigation is open to some possible future directions. A natural
step will be to extend the syntax of terms and types with an
exponential modality. The generalisation of the token machine to this
more expressive language would be an interesting and technically
challenging subject. Moreover, giving a formal status to the
connection between wave-style and the presence of entanglement is a
fascinating subject which we definitely aim to investigate
further. Finally, an interesting proof-theoretical investigation would
consist in analysing the possible connections the with the deep
inference-oriented graph formalism developed in~\cite{deep}.

\bibliographystyle{eptcs}
{\scriptsize
\nocite{VVZ14}
\bibliography{biblio}}
\end{document}

%% file: macros.tex
\usepackage{mathrsfs}
\usepackage{times}
\usepackage{color}
\usepackage{graphics}
\usepackage{pdfsync} 
\usepackage{cmll}
\usepackage{euscript}
\usepackage{amssymb,amsmath}

\input{prooftree}

\newcommand{\comment}[1]{}
\newcommand{\ket}[1]{|#1\rangle}

\newtheorem{definition}{Definition}
\newtheorem{example}{Example}
\newtheorem{proposition}{Proposition}
\newtheorem{lemma}{Lemma}
\newtheorem{theorem}{Theorem}
\newenvironment{proof}{\begin{trivlist}
       \item[\hskip \labelsep {\bfseries Proof.}]}{\hfill $\Box$ \end{trivlist}}

\newcommand{\typet}[2]{#1:#2}
\newcommand{\jd}[2]{#1\vdash #2}

\newcommand{\permone}{\sigma}

\newcommand{\pair}[2]{\langle #1,#2\rangle}
\newcommand{\pairtens}[2]{#1\otimes #2}
\newcommand{\linmap}{\multimap}
\newcommand{\abstr}[2]{\lambda #1.#2}

\newcommand{\pof}{\triangleright}

\newcommand{\QL}{\mathsf{Q}\Lambda}
\newcommand{\varone}{x}
\newcommand{\vartwo}{y}

\newcommand{\patone}{P}
\newcommand{\csone}{C}
\newcommand{\termone}{M}
\newcommand{\termtwo}{N}
\newcommand{\termthree}{L}

\newcommand{\cbitone}{B}
\newcommand{\bitone}{b}

\newcommand{\conone}{\Gamma}
\newcommand{\contwo}{\Delta}

\newcommand{\emcon}{\cdot}
\newcommand{\typeone}{A}
\newcommand{\typetwo}{B}
\newcommand{\typethree}{C}

\newcommand{\unopone}{\mathbf{U}}

\newcommand{\opone}{U}

\newcommand{\subst}[3]{#1\{#2/#3\}}
\newcommand{\eqterm}{\approx}
\newcommand{\eqtermnf}{\sim}
\newcommand{\uset}{\mathscr{U}}

\newcommand{\IAM}[1]{\mathsf{IAM}_{#1}}
\newcommand{\autom}[1]{\mathcal{A}_{#1}}

\newcommand{\states}[1]{\mathcal{S}_{#1}}
\newcommand{\trans}[1]{\rightarrow_{#1}}
\newcommand{\occone}{O}
\newcommand{\occtwo}{P}
\newcommand{\occthree}{R}
\newcommand{\soccone}{\varphi}
\newcommand{\socctwo}{\psi}

\newcommand{\ctone}{C}
\newcommand{\cttwo}{D}
\newcommand{\emct}{[\cdot]}
\newcommand{\pconone}{P}
\newcommand{\pcontwo}{Q}
\newcommand{\nconone}{N}
\newcommand{\ncontwo}{M}
\newcommand{\qrone}{\mathbf{Q}}
\newcommand{\qrtwo}{\mathbf{R}}
\newcommand{\qrthree}{\mathbf{S}}

\newcommand{\poccs}[1]{\mathcal{P}(#1)}
\newcommand{\noccs}[1]{\mathcal{N}(#1)}
\newcommand{\emseq}{\varepsilon}

\newcommand{\tdone}{\pi}
\newcommand{\tdtwo}{\rho}
\newcommand{\tdthree}{\sigma}

\newcommand{\sutone}{\mathcal{T}}
\newcommand{\suttwo}{\mathcal{S}}
\newcommand{\sutthree}{\mathcal{V}}

\newcommand{\mapone}[1]{( #1)^{\bullet}}

\newcommand{\MLL}{\textsf{MLL}}
\newcommand{\MELL}{\textsf{MELL}}
\newcommand{\QMLL}{\textsf{QMLL}}
\newcommand{\atomone}{\alpha}
\newcommand{\atomtwo}{\beta}
\newcommand{\formone}{A}
\newcommand{\formtwo}{B}

\newcommand{\lneg}[1]{#1^{\bot}}
\newcommand{\dneg}[1]{#1^{\bot\bot}}

\newcommand{\mll}{\mathsf{MLL}}
\newcommand{\stone}{\mathsf{S}}
\newcommand{\sttwo}{\mathsf{R}}
\newcommand{\stthree}{\mathsf{T}}
\newcommand{\stfour}{\mathsf{U}}

\newcommand{\judgone}{J}
\newcommand{\pmllone}{\xi}
\newcommand{\pmlltwo}{\mu}
\newcommand{\maptwo}[1]{\mathscr{I}(#1)}

\newcommand{\bitocc}[1]{\mathcal{B}(#1)}
\newcommand{\bitval}[1]{\mathcal{V}(#1)}
\newcommand{\transmll}[1]{\mapsto_{#1}}

\newcommand{\mapthree}[3]{\mathscr{R}_{#1,#2}(#3)}
\newcommand{\mapthreenp}[2]{\mathscr{R}_{#1,#2}}

\newcommand{\mmllone}[1]{\mathcal{M}_{#1}}
\newcommand{\stgen}{\mathsf{S}}
\newcommand{\statesmll}[1]{\mathsf{T}_{#1}}

\newcommand{\pfun}[1]{[#1]}
\newcommand{\vecone}{x}

\newcommand{\midd}{\; \; \mbox{\Large{$\mid$}}\;\;}

\newenvironment{varitemize}
{
\begin{list}{\labelitemi}
{\setlength{\itemsep}{0pt}
 \setlength{\topsep}{0pt}
 \setlength{\parsep}{0pt}
 \setlength{\partopsep}{0pt}
 \setlength{\leftmargin}{15pt}
 \setlength{\rightmargin}{0pt}
 \setlength{\itemindent}{0pt}
 \setlength{\labelsep}{5pt}
 \setlength{\labelwidth}{10pt}
}}
{
 \end{list} 
}

\newcounter{numberone}

\newcounter{numbertwo}

\newcommand{\urule}[3]{%
  \prooftree #1 \justifies #2 \using #3 \endprooftree}
\newcommand{\brule}[4]{%
  \prooftree #1\ \ \ #2 \justifies #3 \using #4 \endprooftree}

\newcommand{\tens}{\otimes} 
\newcommand{\lpar}{\parr} 

\newcommand{\VV}{\mathbb{V}}
\newcommand{\AAA}{\mathbb{A}}
\newcommand{\BB}{\mathbb{B}}
\newcommand{\NN}{\mathbb{N}}
\newcommand{\CC}{\mathbb{C}}

%% file: prooftree.tex
\newdimen\proofrulebreadth \proofrulebreadth=.05em
\newdimen\proofdotseparation \proofdotseparation=1.25ex
\newdimen\proofrulebaseline \proofrulebaseline=2ex
\newcount\proofdotnumber \proofdotnumber=3
\let\then\relax
\def\hfi{\hskip0pt plus.0001fil}
\mathchardef\squigto="3A3B
\newif\ifinsideprooftree\insideprooftreefalse
\newif\ifonleftofproofrule\onleftofproofrulefalse
\newif\ifproofdots\proofdotsfalse
\newif\ifdoubleproof\doubleprooffalse
\let\wereinproofbit\relax
\newdimen\shortenproofleft
\newdimen\shortenproofright
\newdimen\proofbelowshift
\newbox\proofabove
\newbox\proofbelow
\newbox\proofrulename
\def\shiftproofbelow{\let\next\relax\afterassignment\setshiftproofbelow\dimen0 }
\def\shiftproofbelowneg{\def\next{\multiply\dimen0 by-1 }%
\afterassignment\setshiftproofbelow\dimen0 }
\def\setshiftproofbelow{\next\proofbelowshift=\dimen0 }
\def\setproofrulebreadth{\proofrulebreadth}

\def\prooftree{
\ifnum  \lastpenalty=1
\then   \unpenalty
\else   \onleftofproofrulefalse
\fi
\ifonleftofproofrule
\else   \ifinsideprooftree
        \then   \hskip.5em plus1fil
        \fi
\fi
\bgroup
\setbox\proofbelow=\hbox{}\setbox\proofrulename=\hbox{}%
\let\justifies\proofover\let\leadsto\proofoverdots\let\Justifies\proofoverdbl
\let\using\proofusing\let\[\prooftree
\ifinsideprooftree\let\]\endprooftree\fi
\proofdotsfalse\doubleprooffalse
\let\thickness\setproofrulebreadth
\let\shiftright\shiftproofbelow \let\shift\shiftproofbelow
\let\shiftleft\shiftproofbelowneg
\let\ifwasinsideprooftree\ifinsideprooftree
\insideprooftreetrue
\setbox\proofabove=\hbox\bgroup$\displaystyle 
\let\wereinproofbit\prooftree
\shortenproofleft=0pt \shortenproofright=0pt \proofbelowshift=0pt
\onleftofproofruletrue\penalty1
}

\def\eproofbit{
\ifx    \wereinproofbit\prooftree
\then   \ifcase \lastpenalty
        \then   \shortenproofright=0pt  
        \or     \unpenalty\hfil         
        \or     \unpenalty\unskip       
        \else   \shortenproofright=0pt  
        \fi
\fi
\global\dimen0=\shortenproofleft
\global\dimen1=\shortenproofright
\global\dimen2=\proofrulebreadth
\global\dimen3=\proofbelowshift
\global\dimen4=\proofdotseparation
\global\count255=\proofdotnumber
$\egroup  
\shortenproofleft=\dimen0
\shortenproofright=\dimen1
\proofrulebreadth=\dimen2
\proofbelowshift=\dimen3
\proofdotseparation=\dimen4
\proofdotnumber=\count255
}

\def\proofover{
\eproofbit 
\setbox\proofbelow=\hbox\bgroup 
\let\wereinproofbit\proofover
$\displaystyle
}%
\def\proofoverdbl{
\eproofbit 
\doubleprooftrue
\setbox\proofbelow=\hbox\bgroup 
\let\wereinproofbit\proofoverdbl
$\displaystyle
}%
\def\proofoverdots{
\eproofbit 
\proofdotstrue
\setbox\proofbelow=\hbox\bgroup 
\let\wereinproofbit\proofoverdots
$\displaystyle
}%
\def\proofusing{
\eproofbit 
\setbox\proofrulename=\hbox\bgroup 
\let\wereinproofbit\proofusing
\kern0.3em$
}

\def\endprooftree{
\eproofbit 
  \dimen5 =0pt
\dimen0=\wd\proofabove \advance\dimen0-\shortenproofleft
\advance\dimen0-\shortenproofright
\dimen1=.5\dimen0 \advance\dimen1-.5\wd\proofbelow
\dimen4=\dimen1
\advance\dimen1\proofbelowshift \advance\dimen4-\proofbelowshift
\ifdim  \dimen1<0pt
\then   \advance\shortenproofleft\dimen1
        \advance\dimen0-\dimen1
        \dimen1=0pt
        \ifdim  \shortenproofleft<0pt
        \then   \setbox\proofabove=\hbox{%
                        \kern-\shortenproofleft\unhbox\proofabove}%
                \shortenproofleft=0pt
        \fi
\fi
\ifdim  \dimen4<0pt
\then   \advance\shortenproofright\dimen4
        \advance\dimen0-\dimen4
        \dimen4=0pt
\fi
\ifdim  \shortenproofright<\wd\proofrulename
\then   \shortenproofright=\wd\proofrulename
\fi
\dimen2=\shortenproofleft \advance\dimen2 by\dimen1
\dimen3=\shortenproofright\advance\dimen3 by\dimen4
\ifproofdots
\then
        \dimen6=\shortenproofleft \advance\dimen6 .5\dimen0
        \setbox1=\vbox to\proofdotseparation{\vss\hbox{$\cdot$}\vss}%
        \setbox0=\hbox{%
                \advance\dimen6-.5\wd1
                \kern\dimen6
                $\vcenter to\proofdotnumber\proofdotseparation
                        {\leaders\box1\vfill}$%
                \unhbox\proofrulename}%
\else   \dimen6=\fontdimen22\the\textfont2 
        \dimen7=\dimen6
        \advance\dimen6by.5\proofrulebreadth
        \advance\dimen7by-.5\proofrulebreadth
        \setbox0=\hbox{%
                \kern\shortenproofleft
                \ifdoubleproof
                \then   \hbox to\dimen0{%
                        $\mathsurround0pt\mathord=\mkern-6mu%
                        \cleaders\hbox{$\mkern-2mu=\mkern-2mu$}\hfill
                        \mkern-6mu\mathord=$}%
                \else   \vrule height\dimen6 depth-\dimen7 width\dimen0
                \fi
                \unhbox\proofrulename}%
        \ht0=\dimen6 \dp0=-\dimen7
\fi
\let\doll\relax
\ifwasinsideprooftree
\then   \let\VBOX\vbox
\else   \ifmmode\else$\let\doll=$\fi
        \let\VBOX\vcenter
\fi
\VBOX   {\baselineskip\proofrulebaseline \lineskip.2ex
        \expandafter\lineskiplimit\ifproofdots0ex\else-0.6ex\fi
        \hbox   spread\dimen5   {\hfi\unhbox\proofabove\hfi}%
        \hbox{\box0}%
        \hbox   {\kern\dimen2 \box\proofbelow}}\doll%
\global\dimen2=\dimen2
\global\dimen3=\dimen3
\egroup 
\ifonleftofproofrule
\then   \shortenproofleft=\dimen2
\fi
\shortenproofright=\dimen3
\onleftofproofrulefalse
\ifinsideprooftree
\then   \hskip.5em plus 1fil \penalty2
\fi
}
